%% file: root.tex
\documentclass[conference]{IEEEtran}
\usepackage{times}

\input{notation}

\usepackage[numbers]{natbib}
\usepackage{multicol}
\usepackage[bookmarks=true]{hyperref}

\usepackage{graphics} 
\usepackage{epsfig} 
\usepackage{color}
\usepackage{mathptmx} 
\usepackage{times} 
\usepackage{amsmath} 
\usepackage{amssymb}  
\usepackage{todonotes}
\usepackage{graphicx}
\usepackage{mathtools}
\usepackage{dirtytalk}
\usepackage{hyperref}
\usepackage{multicol}
\usepackage[caption=false]{subfig}

\begin{document}

\title{Eyes-Closed Safety Kernels: \\ 
Safety of Autonomous Systems Under Loss of Observability}



\author{

\authorblockN{
Forrest Laine, Chih-Yuan Chiu, and Claire Tomlin
}
}
\maketitle

\begin{abstract}
A framework is presented for handling a potential loss of observability of a dynamical system in a provably-safe way. Inspired by the fragility of data-driven perception systems used by autonomous vehicles, we formulate the problem that arises when a sensing modality fails or is found to be untrustworthy during autonomous operation. We cast this problem as a differential game played between the dynamical system being controlled and the external system factor(s) for which observations are lost. The game is a zero-sum Stackelberg game in which the controlled system (leader) is trying to find a trajectory which maximizes a function representing the safety of the system, and the unobserved factor (follower) is trying to minimize the same function. The set of winning initial configurations of this game for the controlled system represent the set of all states in which safety can be maintained with respect to the external factor, even if observability of that factor is lost. This is the set we refer to as the Eyes-Closed Safety Kernel. In practical use, the policy defined by the winning strategy of the controlled system is only needed to be executed whenever observability of the external system is lost or the system deviates from the Eyes-Closed Safety Kernel due to other, non-safety oriented control schemes. We present a means for solving this game offline, such that the resulting winning strategy can be used for computationally efficient, provably-safe, online control when needed. The solution approach presented is based on representing the game using the solutions of two Hamilton-Jacobi partial differential equations. We illustrate the applicability of our framework by working through a realistic example in which an autonomous car must avoid a dynamic obstacle despite potentially losing observability.
\footnote{This research is supported by the DARPA Assured Autonomy program.\\ All authors are with the Department of Electrical Engineering and Computer Sciences at UC Berkeley. Contact info: \texttt{forrest.laine@berkeley.edu}, \texttt{chihyuan\_chiu@berkeley.edu}, \texttt{tomlin@eecs.berkeley.edu}}
\end{abstract}

\IEEEpeerreviewmaketitle


\section{Introduction} \label{sec:intro}

Before deploying autonomous systems into uncontrolled environments, a degree of confidence in the ability of those systems to operate safely and proficiently is obviously needed. This need is exacerbated for safety-critical systems that can become dangerous in their failure modes.  One way to generate confidence in a system is to prove safety properties for a model of the system. Although such an exercise is imperfect---the modeling assumptions on the system may not hold in practice---the space over which the assumptions do hold becomes a validated operating region in which the performance of the system is guaranteed.   

In this article, we present a method for generating a control scheme for which the safety of the modeled system can be proven, without assuming the \textit{persistence of observability} of the system. In system analysis, an observation model for the system is assumed. Common variants of observation models include full observations (all state dimensions of the system are observed), partial observations (only some state dimensions are observed), and noisy observations (observations are corrupted by noise). Typically, however, these observation models are assumed to be constant throughout the mission of the autonomous system. We relax that assumption to allow for cases in which the observability of the system (or part of the system) might suddenly be lost. 

\begin{figure}[t] 
  \centering
  \smallskip
  \includegraphics[scale=0.125]{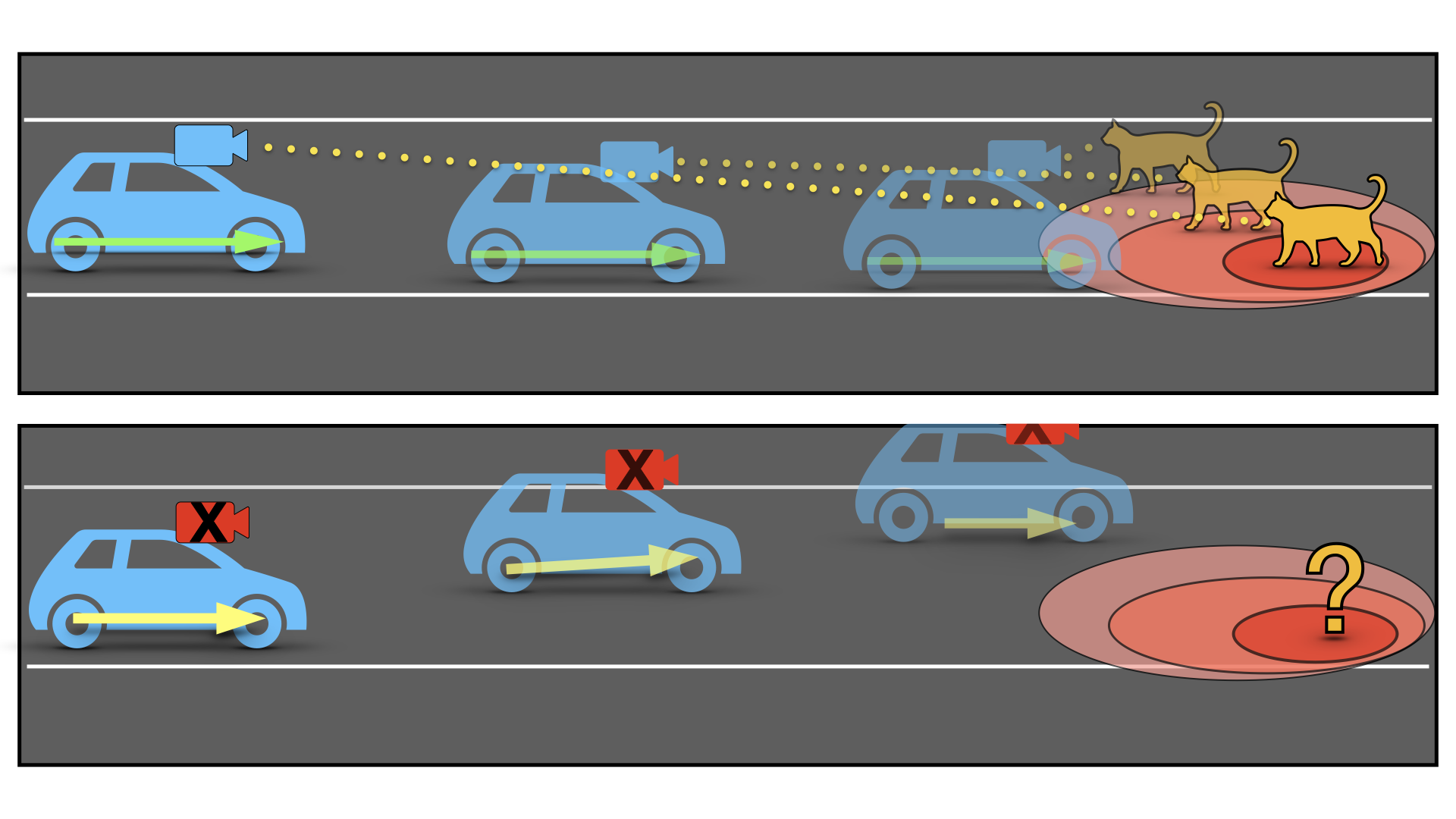} 
  \caption{Graphic Depiction of the Loss of Observability Situation. When planning avoidance maneuvers, the vehicle must assume what type of observations it will receive of the cat in the future. If it assumes it will continue to observe the cat for the duration of the trajectory, it can plan a maneuver which relies on the ability to sense and avoid the cat. This situation is depicted in the top pane of the graphic. At the time the vehicle reaches the location of the cat, it avoids the cat through a region that the cat \textit{could have been in}, but was detected to not be in. Alternatively, if the vehicle assumes it will not receive further observations, or might lose them at any time, its avoidance trajectories will be distinctly different. These maneuvers must avoid all regions in which the cat might be, because the vehicle cannot observe it and determine its exact location. The trajectory in the top pane might be preferred, but might not be safe if observations are lost halfway through the maneuver. The difference between these two models is subtle but important. The design of controllers which are safe with respect to this latter paradigm is the topic of this article.}
  \label{figure:comic}
\end{figure}

Losses of system observability have always been an important consideration, as sensors can fail for various reasons, but accounting for these losses has become increasingly necessary as the perception modules used in modern autonomous vehicles have transitioned to data-driven components. It is very hard to verify that these components will always function as intended, especially if the input data encountered at run-time is dissimilar to the data for which those systems were trained on. 
This poses a significant issue for safety-critical systems. We would like to be able to leverage the advances and strengths of data-driven techniques, yet we desire that the autonomous systems that we develop are still provably safe. 

One approach to developing such systems might derive from \textit{anomaly detection}, which aims to predict whether the input data is anomalous with respect a previously-verified data set, and therefore has a high risk of causing the sensing modality to fail. This is potentially an easier problem than guaranteeing correctness over geometric regions of the data space. If it is accepted that the sensing modality can fail unexpectedly during operation, care must be taken so that the autonomous system is able to maintain safe and proficient operation despite such failures potentially happening during dynamic situations.


The purpose and contribution of this article is the design of a control scheme which can guarantee safe handling of a detected loss of observability of the autonomous system or factors it interacts with. We design such a controller by formulating and solving a differential game representing the situation of interest. We do \textit{not} address the challenge of correctly detecting anomalous inputs to the data-driven algorithms (or other predictions of failure), and instead point to the growing literature of methods interested in addressing that problem, such as \cite{balakrishnan2019hypothesis}, \cite{liang2017enhancing}, \cite{hendrycks2016baseline}, \cite{hendrycks2018deep}, \cite{chalapathy2019deep}, \cite{choi2018generative} to list just a few. The framework presented here is complementary to those methods, offering a mathematically sound means of handling such detected anomalies in the context of autonomous systems. 

A graphical overview of the problem formulation is depicted in Figure \ref{figure:comic}. We use this graphic to reiterate the key difference between the method presented here and that of the standard safe obstacle avoidance paradigms. Typical formulations of obstacle avoidance can similarly be cast as a differential game such as in \cite{fisac2015reach}, although those and other formulations assume (to the best of our knowledge) that the observation model is persistent. \textbf{Under the standard assumption, avoidance trajectories are planned assuming that the system will be able to sense and react to obstacles in the future. This is in contrast to the framework we present, where avoidance trajectories are generated assuming that \emph{observations will potentially be lost at any instant}}. Again, to the best of our knowledge, this is the first work to account for such situations in a provably safe way.  

Before presenting this framework in detail, we summarize the contributions we make: 
\begin{itemize}
\item We develop a mathematical formulation of the loss of observability scenario as a zero-sum, differential Stackelberg game.
\item We develop a novel means of computing a global solution to the game by showing how it can be equivalently expressed through two related partial differential equations. We derive this alternative representation and demonstrate how the global solution can be solved for numerically. 
\item We work through an example demonstrating the practical applicability of our approach to robotic systems. Particularly, we indicate how an autonomous vehicle can (in real time) avoid collision with a dynamic obstacle when an unreliable perception module is responsible for tracking it. 
\end{itemize}

\section{Formulation} \label{sec:formulation}
We begin presentation of our framework by first introducing a model of the types of autonomous systems we wish to consider. In this section and the remainder of this article, we refer to this model as the system of consideration itself, although it should be kept in mind that all properties stated are either modeling assumptions or properties derived from the model. 

To represent the problem of interest, we assume that the full state of the autonomous system can be decomposed into two independent sub-systems, which we refer to as the internal and external systems. We assume that the internal and external systems have independent controls and state evolutions, which depend only on their respective states and controls. We assume that both systems are able to observe their own state at all times, that the external system can observe the internal system at all times, and that the internal system can observe the external system in an impersistent manner, meaning observations can be lost and regained at arbitrary times. We take the perspective of the internal system, and aim to control it in a manner that maintains a safety specification involving both systems, robust to potentially  impersistent observations of and adversarial control by the external system. 

An example of a system with such a decomposition can be seen in Figure \ref{figure:comic}, in which the vehicle is the internal system and the obstacle (cat) is the external system. There, the steering rate and acceleration of the vehicle might be the internal controls, and the directional velocity of the cat might be the external control. Because the vehicle is tracking the cat using an unreliable perception module, the observations of the cat are assumed to be impersistent. 

Assume without loss of generality that time begins at some $t_0 \geq 0$, at which point the state of both systems are known to both systems, within some bounded error tolerance. 
Define $x_{int} \in \R^{n_{int}}$ and $x_{ext} \in \R^{n_{ext}}$ as the internal and external states, respectively, evolving according to the dynamics:
\begin{align} \label{Eqn: Dynamics of Pursuer and Evader}
    \dot{x}_{int}(t) &= f_{int}(x_{int}(t), u_{int}(t)), &t \in [t_0, t_f] \\ \nonumber
    \dot{x}_{ext}(t) &= f_{ext}(x_{ext}(t), u_{ext}(t)), &t \in [t_0, t_f]
\end{align}
with control signals drawn from the set of measurable functions:
\begin{equation}
    \begin{split}
        u_{int}(\cdot) \in \mathbb{U}_{int} = \{\phi_{int}: \hspace{0.5mm} &[t_0, t_f] \ra \mathcal{U}_{int}: \\
        &\phi_{int} \text{ is measurable} \}, \\
        u_{ext}(\cdot) \in \mathbb{U}_{ext} = \{\phi_{ext}: \hspace{0.5mm} &[t_0, t_f] \ra \mathcal{U}_{ext}: \\
        &\phi_{ext} \text{ is measurable} \}.
    \end{split}
\end{equation}
where $\mathcal{U}_{int} \subset \R^{m_{int}}$ and $\mathcal{U}_{ext} \subset \R^{m_{ext}}$ are compact subsets of the input spaces $\R^{m_{int}}, \R^{m_{ext}}$. 

We assume that the function representing the dynamics of the internal system, $f_{int}: \R^{n_{int}} \times \mathcal{U}_{int} \ra \R^{n_{int}}$, is uniformly continuous, bounded, and Lipschitz continuous in its first argument. As such, for every initial external system state $x_{int0} \in \R^{n_{int}}$ and internal control signal $u_{int}(\cdot) \in \mathbb{U}_{int}$, then as stated in \cite{EarlA.Coddington1955}, there exists a unique trajectory $\xi_{int}(t; t_0, x_{int0}, u_{int}(\cdot))$ satisfying:
\begin{align} 
\begin{split}
\frac{d}{dt} \xi_{int}(t; t_0, x_{int0}, u_{int}(\cdot)) &= \\
 f_{int}(\xi_{int}&(t; t_0, x_{int0}, u_{int}(\cdot)), u_{int}(t)),
\end{split}  \\
\xi_{int}(t_0; t_0, x_{int0}, u_{int}(\cdot)) &= x_{int0}.
\end{align}


Similarly, define the inverse trajectory of the internal system as follows. Let $\xi_{int}^{-1}$ be the trajectory which propagates states forward in time starting from some state $x^{-1}_{int0}$ at $t_0$ up until time $t$, using the \textit{reverse-time dynamics} of the internal system:
\begin{align}
    \begin{split}
	\frac{d}{ds} \xi_{int}^{-1}(t; t_0, x^{-1}_{int0}, u_{int}(\cdot)) &= \\
    - f_{int}(& \xi_{int}^{-1}(t; t_0, x^{-1}_{int0}, u_{int}(\cdot)), u_{int}(t)), \end{split} \\
    \xi_{int}(t_0; t_0, x^{-1}_{int0}, u_{int}(\cdot)) &= x^{-1}_{int0}.
\end{align}
With a slight abuse of notation, $x^{-1}_{int0}$ is used to indicate the initial state of the reverse-time trajectory mapping $\xi^{-1}_{int}$, and is not related to the variable $x_{int0}$ in any direct way.

The same assumptions are made on the dynamics of the external system, such that state trajectories of the internal system are defined analogously of those of the internal system, and given by $\xi_{ext}(t;t_0,x_{ext0},u_{ext}(\cdot))$ and $\xi^{-1}_{ext}(t;t_0,x^{-1}_{ext0},u_{ext}(\cdot))$.

To formally define the impersistent observation model on the external system, we define $\Sigma$ to be a $\sigma$-algebra on $[0,\infty)$, and let $\mathbf{T} \in \Sigma$ be a set of potentially open and infinite intervals on $[0,\infty)$, \textbf{unknown to the internal system}, which denote periods of time in which the internal system's observations of the external system are lost. Define $y_{int}(t) \in \R^{n_{int}} \times \{\R^{n_{ext}} \cup \varnothing \}$ to be the observation of the internal system at time $t$, with $\varnothing$ representing the $n_{ext}$-dimensional empty value. The observation model is then given by
\begin{equation}
	y_{int}(t) = 
	\begin{cases} 
		 \begin{bmatrix} x_{int}^\intercal(t),   &x_{ext}^\intercal(t) \end{bmatrix}^\intercal   \ &t \notin \mathbf{T} \\
		 \begin{bmatrix} x_{int}^\intercal(t),  & \varnothing^\intercal \ \ \ \ \end{bmatrix}^\intercal   &t \in \mathbf{T}
	\end{cases}
\end{equation}
Conversely, the observation model for the external system is simply given by 
\begin{equation}
	y_{ext}(t) = \begin{bmatrix} x_{int}^\intercal(t),  & x_{ext}^\intercal(t) \end{bmatrix}^\intercal, \  \forall t \in [t_0, t_f].
\end{equation}
These definitions formalize the notion that for times $t\in\mathbf{T}$, the internal system can only observe its own state, and not the external state. 

Finally, we define a quantitative notion of safety of the internal system, which is given by the function $d: \R^{n_{int}} \times \R^{n_{ext}} \ra \R$. The joint internal and external state of the system at time $t$ is defined to be safe (for the internal system) if 
\begin{equation} d(x_{int}(t), x_{ext}(t)) \geq 0. 
\end{equation}
Therefore, we can think of the function $d$ as some signed distance to the set of unsafe system configurations.

The objective of the internal system is to  remain safe for all $t \in [t_0, t_f]$ and all possible $\mathbf{T}$. We can, without loss of generality, assume that $\mathbf{T} = (t_0, t_f]$, since any control policy that is safe without observations of the external system will also be safe with those observations, whereas the converse is in general not true. Under this assumption, $\mathbf{T}$ does not include the point $t_0$, since the initial state of each systems is assumed to be known to both systems at $t_0$. 

The assumption that $\mathbf{T} = (t_0, t_f]$ may appear overly-conservative, since in practice the internal system is at least intended to receive observations of the external system. However, if we wish the internal system to truly be robust to the possibility of lack of observations at any and all times, we must assume such is true in the analysis and design of our system. Therefore this choice does not introduce additional conservatism into the formulation. During the actual execution of the system, any received observations should and will be used to allow proficient behavior of the system, which we discuss in Section \ref{sec:system_execution}.

Given the assumptions on $\mathbf{T}$, the function indicating whether the initial internal and external system configurations will be safe for all time $[t_0,t_f]$ can be stated as the sign of the following value function: 
\begin{equation} \label{eq:opt}
\begin{aligned}
	V_{int}(x_{int0}, x_{ext0}, t_0, t_f) &= \\
	\sup_{u_{int}(\cdot) \in \mathbb{U}_{int}} \  & \inf_{u_{ext}(\cdot) \in  \mathbb{U}_{ext}} \  \min_{t \in [t_0,t_f]}  d(x_{int}(t), x_{ext}(t)) \\
	\text{s.t.     }  &x_{int}(t) = \xi_{int}(t; t_0, x_{int0}, u_{int}(\cdot)) \\
	 &x_{ext}(t) = \xi_{ext}(t; t_0, x_{ext0}, u_{ext}(\cdot))
\end{aligned}
\end{equation}
This function computes the least signed-distance encountered to unsafe configurations along the worst-case trajectory of the unobserved external system, and best-effort trajectory of the internal system. A non-negative value function indicates that the internal system can remain safe, despite failure to observe the external system and its adversarial attempts to render the internal system unsafe. A negative value, meanwhile, indicates that there exists control actions by the external system such that no control signal of the internal system can maintain safety for all time, at least without observations of the external system. We wish to compute the super-zero-level set of (\ref{eq:opt}), or equivalently compute the exact set of all safe initial configurations of the system. 

Note that in (\ref{eq:opt}), the optimization for both systems' control signals are \textit{open-loop}, i.e. neither is an explicit function of the state of either system (besides the initial states, $x_{int0}$ and $x_{ext0}$). This is because the internal system is assumed to be unable to observe the state of the external system for all times $t>t_0$, and therefore unable to react to its position during times of unobservability. Since the optimization over the external system control is chosen inside the optimization over internal control signals, this information structure does allow for the external control signal to be chosen knowing the position of the internal system. In this situation, there is no difference between an open-loop and closed-loop control signal for the external system. We note that in game theoretic terms, this game can be described as an open-loop Stackelberg game, in which the internal system is the lead player.

Looking at the structure of (\ref{eq:opt}), it is not obvious how to evaluate the value function for arbitrary inputs, let alone compute the super-zero-level set. This is mainly as it requires solving difficult nested optimizations for each evaluation. This difficulty exists even in offline settings in which real-time evaluation is not necessary. Unlike the case with \textit{closed-loop} differential games, the game value function in this form cannot even be pre-computed using offline dynamic programming, due to the open-loop nature of the information structure. This is because, while the state of the external system is known to the internal system at time $t_0$, for all $t > t_0$, the state of external system is unknown. 

That being said, we will see that because the internal and external system have independent dynamics, (\ref{eq:opt}) can be reformulated in a way that does allow for a dynamic programming solution, allowing its offline computation to leverage standard methods. This alternative formulation will rely on first computing the forward reachable set (FRS) of the external system (set of all states the external system can reach from its initial position), and then computing the backwards avoid tube for the internal system (set of all initial states of the internal system that can avoid the FRS of the external system). Using this method to compute the game value (\ref{eq:opt}) offline, the value and a resultant control policy can be used in online, real-time execution of the system. 

We first formally define the forward reach set of the external system, denoted $\mathbf{R}_{ext}(t, x_{ext0}; t_0)$ for time $t \in [t_0,t_f]$. This is the set of all states $x$ at time $t$ that the external system can reach with admissible control signals, starting from $x_{ext0}$ at time $t_0$. Equivalently, this is the set of all states $x$ at time $t_0$ from which the inverse trajectory of the external system can reach $x_{ext0}$ at time $t_f$.
\begin{equation}
\begin{split}
	\mathbf{R}_{ext}(t, x_{ext0}; t_0) \ \ \ \ \ &  \\
	:= \{ x \in \R^{n_{ext}} & : \exists \  u_{ext}(\cdot) \in \mathbb{U}_{ext} \\
	  \text{  s.t. } & \xi_{ext}(t; t_0, x_{ext0}, u_{ext}(\cdot)) = x\} \\
	\ \ \ \ \ &  \\
	= \{ x \in \R^{n_{ext}} & : \exists \  u_{ext}(\cdot) \in \mathbb{U}_{ext} \\
	  \text{  s.t. } & \xi_{ext}^{-1}(t; t_0, x, u_{ext}(\cdot)) = x_{ext0}\}
\end{split}
\end{equation}
Using the second definition, $\mathbf{R}_{ext}(t,x_{ext0}; t_0)$ can easily be expressed as the sub-zero-level set of the following value function:
\begin{equation} \label{eq:value_r}
\begin{split}
	V_{\mathbf{R}_{ext}}(x,t,x_{ext0};t_0) &:= \inf_{u_{ext}(\cdot) \in \mathbb{U}_{ext}} \|z(t) - x_{ext0}\| \\
	\text{s.t.    } & z(t) = \xi_{ext}^{-1}(t; t_0, x, u_{ext}(\cdot))
\end{split}
\end{equation}


Following immediately from the definition,
\begin{equation} \label{eq:value_r_init}
V_{\mathbf{R}_{ext}}(x, t_0, x_{ext0};t_0) = \|x - x_{ext0}\|.
\end{equation}
 

The following theorem shows that $V_{\mathbf{R}_{ext}}(x,t,x_{ext0};t_0)$ , and consequently $\mathbf{R}_{ext}(t, x_{ext0}; t_0)$, can be numerically solved for all $t \in [t_0, t_f]$ using a variational Hamiltonian-Jacobi partial differential equation. 

\begin{theorem} \label{thm:frs_external}
Suppose $f_{ext}: \R^{n_{ext}} \times \mathcal{U}_{ext} \ra \R^{n_{ext}}$, is uniformly continuous, bounded, and Lipschitz continuous in its first argument. Then the value function $V_{\mathbf{R}_{ext}}(x,t,x_{ext0};t_0)$ defined in (\ref{eq:value_r}) is the unique viscosity solution of the variational Hamilton-Jacobi PDE:
\begin{equation} \label{pde1:eq}
    \begin{split}
        \frac{\partial V_{\mathbf{R}_{ext}}}{\partial t} + &H = 0, \\
        &\forall t \in [t_0, t_f], x \in \R^{n_{ext}},
    \end{split}
\end{equation}
with initial condition:
\begin{equation} \label{pde1:ic}
    V_{\mathbf{R}_{ext}}(x, t_0, x_{ext0};t_0) = \Vert x - x_{ext0} \Vert,
\end{equation}
where the Hamiltonian, $H$, is defined by:
\begin{equation} \label{pde1:ham}
    H = \min_{u \in \mathcal{U}_{ext}} \frac{\partial \mathbf{V}_{R_{ext}}}{\partial x}(x, t, x_{ext0}; t_0) \cdot \left(- f_{ext}(x, u)\right).
\end{equation}
\end{theorem}

\begin{proof}
(see \cite{bansal2017hamilton}, Section II).
\end{proof}

Therefore, this PDE can be solved numerically to obtain the set $\mathbf{R}_{ext}(t,x_{ext0};t_0)$, which can then be used to define a new, \textit{pessimistic} safety function, which eliminates the need to know the exact state $x_{ext}(t)$.
\begin{equation}
\hat{d}(x_{int}, t, x_{ext0}; t_0) := \min_{x_{ext} \in \mathbf{R}_{ext}(t, x_{ext0};t_0)} d(x_{int}, x_{ext}).
\end{equation}
Observe that the safety of the internal system at time $t$ can equivalently be expressed through the condition $\hat{d}(x_{int}(t),t, x_{ext0}; t_0) \geq 0$. We note that this function in general may also be difficult to evaluate, however, as it too requires solving a potentially difficult optimization problem for each evaluation. By defining the set $\mathbf{D}$ to be the sub-zero-level set of the function $\hat{d}$, we can construct an alternative pessimistic safety function that will equivalently represent our safety criterion. First define the set 
 \begin{equation}
 \begin{split}
 \tilde{\mathbf{D}}(t, x_{ext0}; t_0) \subset \R^{n_{int}} &:= \\
 proj_{int}( \mathbf{D} & \cap (\mathbf{R}_{ext}(t,x_{ext0};t_0) \times \R^{n_{int}})).
 \end{split}
 \end{equation}
 Here $proj_{int}$ is the projection operator, projecting arguments into the state space of the internal system. In words, for a given time $t$, and initial state $x_{ext0}$ at initial time $t_0$, the set $\tilde{\mathbf{D}}(t, x_{ext0}; t_0)$ is the set of all states $x_{int}$ such that there exists a $x_{ext}$ reachable by the external system at time $t$, such that the pair $(x_{int}, x_{ext})$ is unsafe. Now we define the alternate pessimistic safety function as the signed distance to this set:
\begin{equation}
	\tilde{d}(x_{int},x_{ext0}, t; t_0) := \begin{cases}
		\min_{x \in \tilde{\mathbf{D}}(t, x_{ext0};t_0)} \| x - x_{int} \|   &x_{int} \notin \tilde{\mathbf{D}} \\
		 -\min_{x \notin \tilde{\mathbf{D}}(t, x_{ext0};t_0)} \| x - x_{int} \|   &x_{int} \in \tilde{\mathbf{D}}
	\end{cases}
\end{equation}
Using this alternate pessimistic safety function, we define a new value function,
\begin{equation} \label{eq:opt2}
\begin{aligned}
    \tilde{V}_{int}(x_{int0}, x_{ext0}, t, t_0, t_f) = & \\
    \sup_{u_{int}(\cdot) \in \mathbb{U}_{int}} & \   \min_{s \in [t,t_f]} \  \tilde{d}(x_{int}(s), s,x_{ext0}; t_0) \\
	\text{s.t.    }  x_{int}(s) &= \xi_{int}(s; t, x_{int0}, u_{int}(\cdot)) 
\end{aligned}
\end{equation}
Note that 
\begin{equation} \label{eq:value_a_int}
\tilde{V}_{int}(x, x_{ext0}, t_f, t_0, t_f) = \tilde{d}(x, t_f, x_{ext0}; t_0). 
\end{equation}
Here, an extra argument on time is present in contrast to the function $V_{int}$ in (\ref{eq:opt}). When evaluated at $t=t_0$, however, the value function $\tilde{V}_{int}$  has an equivalent zero-level set as $V_{int}$. Since we are ultimately interested in the super-zero-level set of $V_{int}$, we can equivalently work with $\tilde{V}_{int}$, which we will show can also be computed with known methods. Given the definition of this new value function, it becomes evident that the super-zero-level set, when evaluated at a particular $x_{ext0}$, is the set of states $x_{int0}$ for which the internal system can avoid the set $\tilde{\mathbf{D}}(t, x_{ext0}; t_0)$, for all $t \in [t_0, t_f]$. This is precisely the definition of backwards avoid-tube of the set $\tilde{\mathbf{D}}(t, x_{ext0}; t_0)$, which we denote by $\mathbf{A}_{int}$:
\begin{equation}
\begin{split}
	\mathbf{A}_{int}(&t, \tilde{D}(\cdot,\cdot;\cdot), x_{ext0}; t_0, t_f) := \\
	\{ & x \in \R^{n_{int}} :  \ \exists u_{int}(\cdot)  \in \mathbb{U}_{ext}, \\
	 & \text{ s.t. }  \xi_{int}(s; t,  x,  u_{int}(\cdot)) \notin \tilde{\mathbf{D}}(s, x_{ext0}; t_0), \ \forall s \in [t, t_f]\}.
\end{split}
\end{equation}
In this definition, the argument $t$ is the time corresponding to the beginning of the internal system trajectory, where as the set $\tilde{\mathbf{D}}$ is always computed with respect to initial time $t_0$. 

Similar to $V_{\mathbf{R}_{ext}}(x,t,x_{ext0};t_0)$, the value function  $\tilde{V}_{int}(x,x_{ext},t,t_0,t_f)$ (for a particular $x_{ext}$, $t_0$, and $t_f$), and hence the avoid tube $\mathbf{A}_{int}(t, \tilde{\mathbf{D}}, x_{ext0}; t_0, t_f)$,  can be solved using a HJ PDE.
\begin{theorem} \label{thm:brt_internal}
Suppose $f_{int}: \R^{n_{int}} \times \mathcal{U}_{int} \ra \R^{n_{int}}$, is uniformly continuous, bounded, and Lipschitz continuous in its first argument. Then the value function $\tilde{V}_{int}(x,x_{ext},t,t_0,t_f)$ is the unique viscosity solution of the variational Hamilton-Jacobi PDE:
\begin{equation} \label{pde2:eq}
    \begin{split}
    0 = \max & \left\{ \frac{\partial \tilde{V}_{int}}{\partial t} + H, \tilde{d}(x, t, x_{ext0}; t_0) - \tilde{V}_{int}(x,x_{ext},t,t_0,t_f) \right\} \\
         & \forall t \in [t_0, t_f], x \in \R^{n_{int}}, 
    \end{split}
\end{equation}
with terminal condition:
\begin{equation} \label{pde2:ic}
    \tilde{V}_{int}(x, x_{ext0}, t_f, t_0, t_f) = \tilde{d}(x, t_f, x_{ext0}; t_0),
\end{equation}
where the Hamiltonian is defined by:
\begin{equation} \label{pde2:ham}
    H = \min_{u \in \mathcal{U}_{int}} \frac{\partial \tilde{V}_{int}}{\partial x}(x,x_{ext},t,t_0,t_f) \cdot f_{int}(x, u).
\end{equation}
\end{theorem}

\begin{proof}
(see \cite{fisac2015reach}, Theorem 1).
\end{proof}

Therefore, by first solving the partial differential equation given by (\ref{pde1:eq}), (\ref{pde1:ic}), (\ref{pde1:ham}), and then the partial differential equation corresponding to (\ref{pde2:eq}), (\ref{pde2:ic}), (\ref{pde2:ham}), we can obtain a value function for which the super-zero-level set denotes the set of all safe states for the internal system, with the initial configuration of the external system parameterized by $x_{ext0}$ and despite un-observability of the external system. We denote this set the Eyes-Closed Safety Kernel of the internal system.
\begin{equation} \label{eq:eyes}
\mathbf{E}(t,t_0,t_f) = \{(x_{int}, x_{ext}) : \tilde{V}_{int}(x_{int},x_{ext},t,t_0,t_f) \geq 0\}
\end{equation}

Here we make an important note. In order to compute $\mathbf{E}(t, t_0,t_f)$, both of the PDEs must be solved for \textit{all} $x_{ext0}$ of interest. However, both PDEs can in general be solved for an entire set of $x_{ext0}$ simultaneously. This formulation would simply result in PDEs which have state dimension equal to $2n_{ext}$ and $(n_{int}+n_{ext})$ for the first and second PDEs, respectively. The equations given in (\ref{pde1:eq}), (\ref{pde1:ic}), (\ref{pde1:ham}), (\ref{pde2:eq}), (\ref{pde2:ic}), and (\ref{pde2:ham}) are still valid in this case, as the states $x_{ext0}$ do not evolve through time. We made explicit the computation of both PDEs as functions of a particular $x_{ext0}$ since in practice, symmetries in the problem can be exploited to avoid computing the full PDEs, which can be expensive or even intractable to compute. We demonstrate this technique in the example in Section \ref{sec:example}.

\subsection{System Execution} \label{sec:system_execution}
We stated that the result of the preceding calculations, the eyes-closed safety kernel, could be used for safe autonomous execution when the external system is unobserved as well as when it is observed. Recall that $\mathbf{T}$ (the intervals of time for which the external system observations are lost) is unknown to the internal system. However, we assume that the internal system is endowed with the ability to know at some $t\in [t_0,t_f]$ if $t\in\mathbf{T}$. In other words, the internal system is able to instantaneously detect that it has lost observations of the external system, through use of an anomaly detection module or by some other means. 

Note that since the dynamics of both the internal and external systems do not depend on time, the value functions computed are temporal-shift invariant. This means that $\mathbf{E}(x_{int}, x_{ext}, 0, 0, t_f)$ provides a safe operating region for which the internal system can operate in, guaranteeing the existence of a safe control signal for $t_f$ seconds. 

Over those $t_f$ seconds, the optimal control for the internal system at state $x_{int}$ at time $t$ with last-known observation of the external system at $x_{ext0}$ is given by 
\begin{equation} \label{eq:opt_ctrl}
\begin{split}
    u_{int}^*(x_{int}, t; x_{ext0}, t_f) &= \\
    \text{arg}\min_{u \in \mathcal{U}_{int}} &  \frac{d\tilde{V}_{int}}{dx_{int}}(x_{int},x_{ext0},t,0,t_f) \cdot f_{int}(x_{int},u)
    \end{split}
\end{equation}

If at any time before the $t_f$ seconds are expired, the internal system detects that $t \notin \mathbf{T}$ and observations of the external system are regained, then the control of the internal system can resume nominal control designed to satisfy some other objective, subject to remaining in the set $\mathbf{E}(0,0,t_f).$ We note that by  definition, there always exist control actions that keep the system within $\mathbf{E}$, regardless of the control of the external system. A practical way to ensure that the system never leaves the set is to monitor the value $\tilde{V}_{int}(x_{int},x_{ext},0,0,t_f)$ and switch to the control defined in (\ref{eq:opt_ctrl}) if the value drops within some tolerance above $0$.  

We further note that the guarantees provided on safety are only valid for $t_f$ seconds after observations are lost. However, by taking $t_f$ very large in the computation of (\ref{eq:eyes}), we can obtain guarantees for arbitrarily long intervals of time. In practice, often the value functions of interest are convergent anyways, in which case they converge to the infinite-horizon solution. When that is the case, the safety guarantees never expire. 

\subsection{Note on Sub-Optimal Approximations} \label{sec:approx}
There may exist systems for which the decomposition of the internal and external systems is not possible, yet consideration of the loss of observation of the external states is still desired. In these cases, the method presented here does not directly apply, although modifications can be made to compute sub-optimal yet conservative approximations of the sets defined. 

One such approximation is to fix a control policy for the internal system, so that the state trajectory of the internal system can be evaluated and incorporated into the objective of the external system. Then, computing the set of initial relative configurations of the internal and external system such that the external system can reach unsafe relative states, an \textit{over approximation of the unsafe relative initial configurations} is defined. By taking the complement of this set as the safe operating region for the internal system, the same guarantees on safety that are made in the exact method hold, although this approximation to the safe operating region may be significantly smaller or even empty compared to the true safe operating region. We note that almost all systems of practical consideration are in fact decomposable in the manner outlined, and therefore do not introduce any additional conservatism. 

\section{Example} \label{sec:example}



To demonstrate the method presented, we work through an example situation that is common to autonomous ground vehicles. The situation we wish to model is the interaction between a vehicle (internal system) and a dynamic obstacle (external system). As is typical for autonomous vehicles, the state of the external dynamic obstacle is estimated by the internal system using sensors. We assume that the internal system receives impersistent observations of the external system, which are exact within some bounded tolerance. This resembles a situation in which the internal system is using a data-driven algorithm to detect and track obstacles, but may encounter and detect out of distribution inputs, at which points in the perception module can not be trusted.

We assume for the sake of visualization that both the vehicle and dynamic obstacle are modeled as Dubins Cars. Therefore, their state vectors are both 3-dimensional, making the sets defined in Section \ref{sec:formulation} easy to visualize. The dimension of the two systems need not be the same, and indeed, higher dimensional systems can be considered to represent more realistic systems (barring some computational limits as discussed in Section \ref{sec:computation}). 


The internal system is considered safe if and only if the two vehicles do not collide. For the sake of simplicity and readability, we assume that each car is spherical with radii equal to $r$. Then the set of safe states is defined to be the super-zero-level set of the function 
\begin{equation}
    d(x_{int}, x_{ext}) := (x_1 - x_4)^2 + (x_2 - x_6)^2 - (2r)^2,
\end{equation}
where the dynamics of the internal system are given by
\begin{equation}
\begin{aligned}
    \dot{x}_1(t) &= u_1(t) \cdot cos(x_3(t)) \\
    \dot{x}_2(t) &= u_1(t) \cdot sin(x_3(t)) \\
    \dot{x}_3(t) &= u_2(t),
\end{aligned}
\end{equation}
and the dynamics of the dynamic obstacle (external system) are given by
\begin{equation}
\begin{aligned}
    \dot{x}_4(t) &= u_3(t) \cdot cos(x_6(t)) \\
    \dot{x}_5(t) &= u_3(t) \cdot sin(x_6(t)) \\
    \dot{x}_6(t) &= u_4(t).
\end{aligned}
\end{equation}

We assume that both agents are limited by the control bounds, 
\begin{equation}
\begin{aligned}
    u_1 \in [0, 4], \ \ u_2 \in &[-1.0, 1.0] \\
    u_3 \in [0, 3], \ \ u_4 \in &[-0.75, 0.75].
\end{aligned}
\end{equation}

\begin{figure}[tb] 
  \centering
  \smallskip
  \includegraphics[scale=0.13]{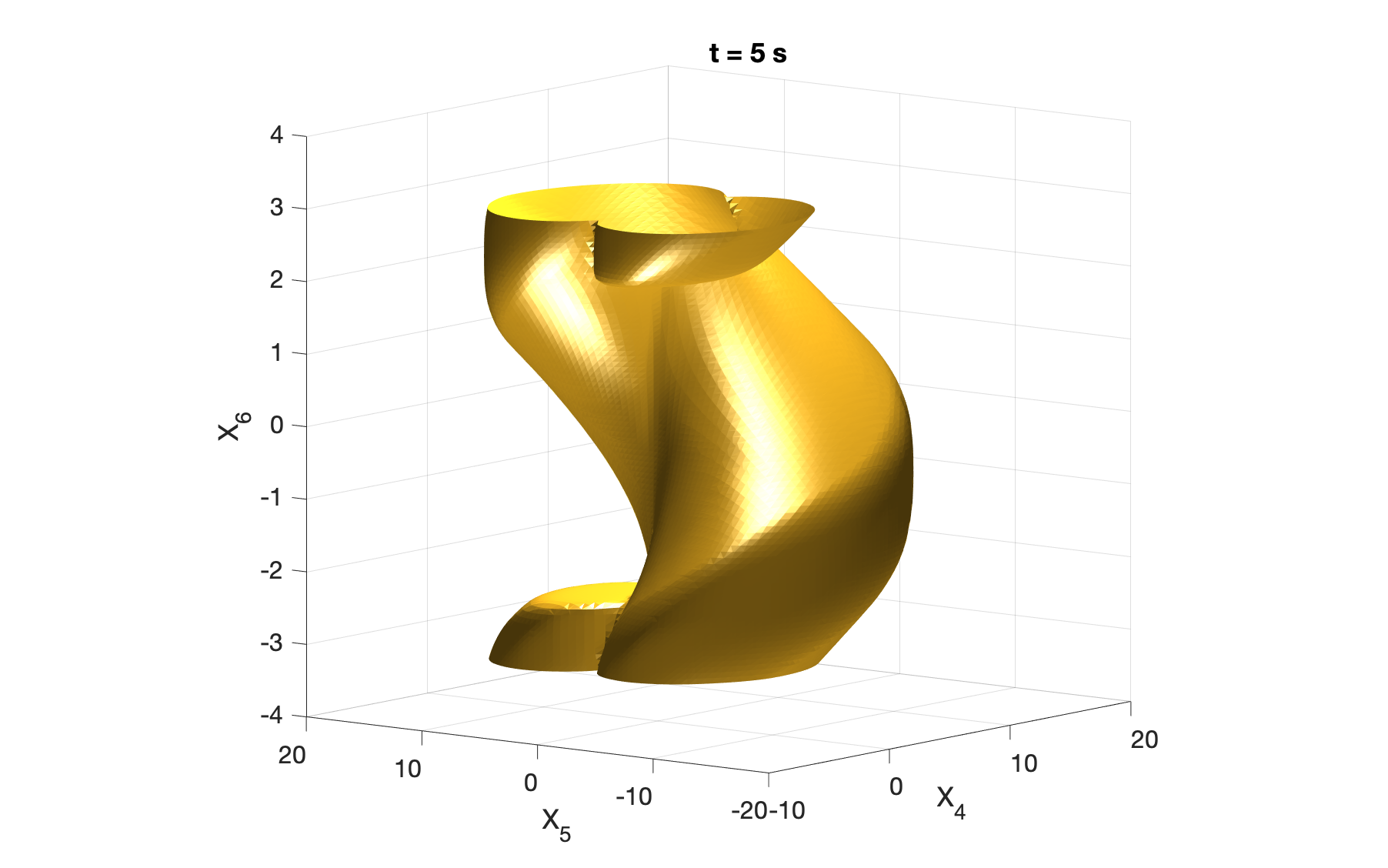} 
  \caption{Visualization of the set $\mathbf{R}_{ext}(5.0,\mathbf{0}; 0)$. The states inside the yellow volume are all the possible states the external system can reach in 5.0 seconds.}
  \label{figure:frs}
\end{figure}

\begin{figure}[tb] 
  \centering
  \smallskip
  \includegraphics[scale=0.13]{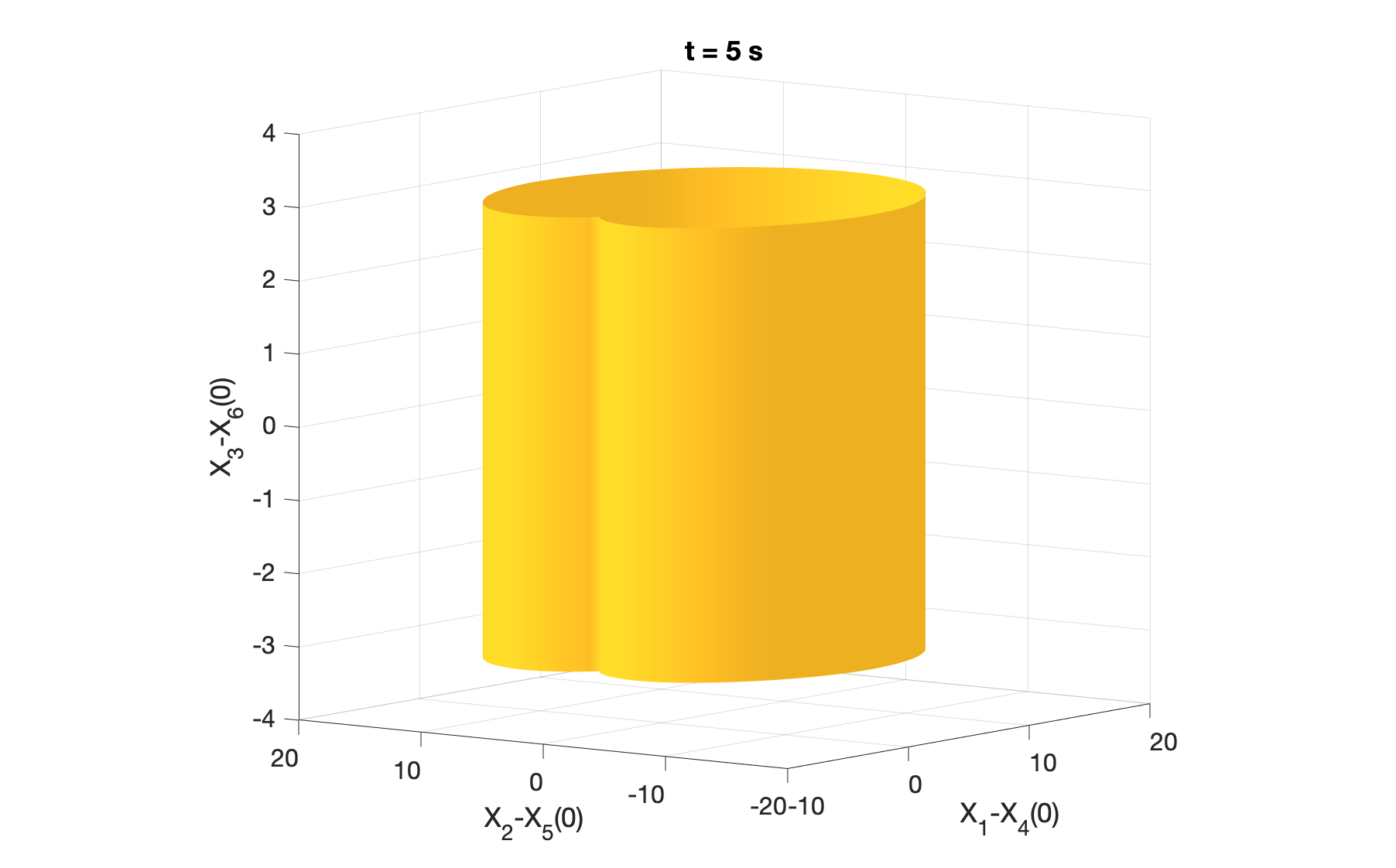} 
  \caption{Visualization of the set $\tilde{\mathbf{D}}(5.0, \mathbf{0}; 0)$. The states inside the yellow volume are all of the states of the internal system, relative to the last known position of the external system ($x_{ext0}$), such that collision at the current time is possible. Note the invariance of the set with respect to the relative heading dimension ($x_3-x_6(0)$). This is because collision should be avoided, irrespective of the heading of the two sub-systems.}
  \label{figure:gtilde}
\end{figure}

\begin{figure}[tb] 
  \centering
  \smallskip
  \includegraphics[scale=0.13]{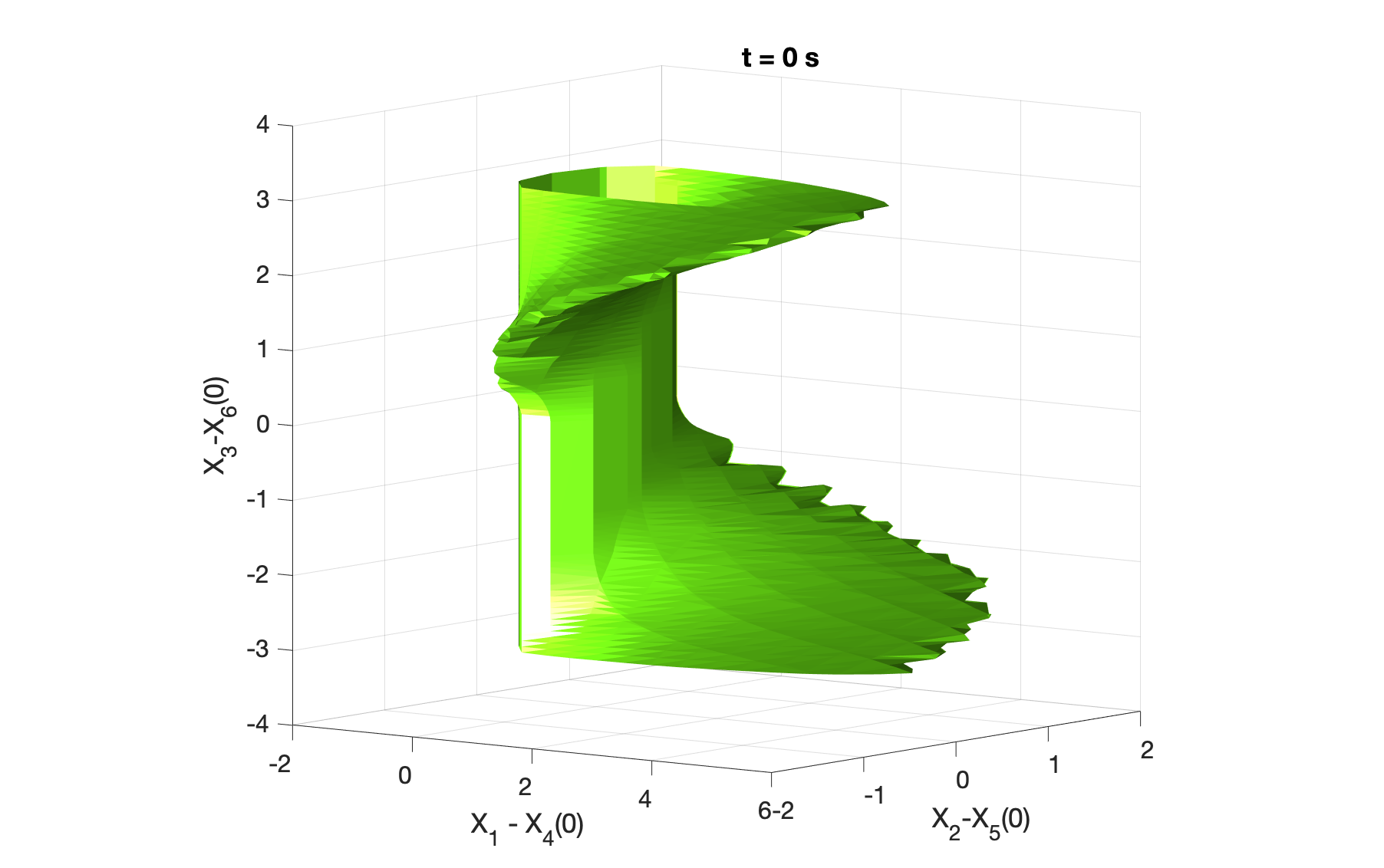} 
  \caption{Visualization of the set $\mathbf{A}_{int}(0, \tilde{\mathbf{D}},\mathbf{0};0,5.0)$. The states inside the green volume are the initial states of the internal system, relative to the initial state of the external system, that must be avoided in order to guarantee collision is avoided for $t_f$ seconds, even in the event that observations of the external system are lost at any time.}
  \label{figure:brt}
\end{figure}

\begin{figure}[tb] 
  \centering
  \smallskip
  \includegraphics[scale=0.22]{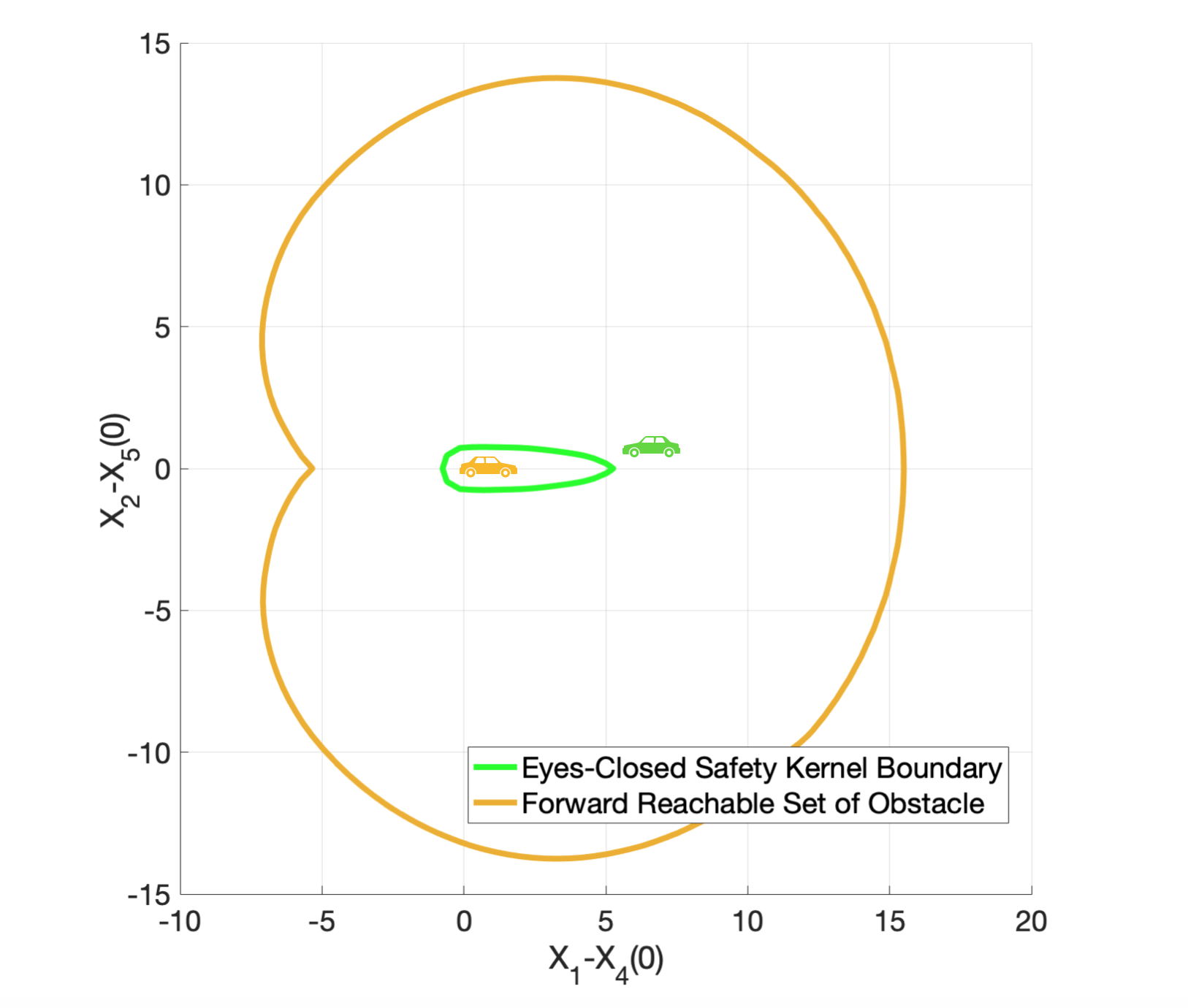} 
  \caption{Visualization of the Eyes-Closed Safety Kernel, $\mathbf{A}_{int}(0, \tilde{\mathbf{D}},\mathbf{0};0,5.0)$, and the Forward Reachable Set of the obstacle, $\tilde{\mathbf{D}}(5.0, \mathbf{0}; 0)$, for the slice $X_3-X_6(0) = \pi$. Here we see that despite the green vehicle (internal system) starting within the FRS of the orange vehicle (external system), the green vehicle lies within the Eyes-Closed Safety Kernel. This means that the dynamics of the combined system are such that there exists a maneuver by the green vehicle which can evade the orange vehicle, despite losing observability.}
  \label{figure:2dslice}
\end{figure}

We solve for the sets $\mathbf{R}_{ext}(t, x_{ext0}; 0)$, $\tilde{\mathbf{D}}(t,x_{ext0};0)$, and $\mathbf{A}_{int}(t, \tilde{\mathbf{D}},x_{ext0};t_0,t_f)$, defined in Section \ref{sec:formulation}. We assume without loss of generality that $t_0 = 0$, and solve for $t_f = 5.0$. We find that this time horizion is sufficient for the set $\mathbf{A}_{int}$ to converge to within an acceptably small tolerance. 

To solve for the set $\mathbf{R}_{ext}(t,x_{ext0}; 0)$ we use the Level Set Toolbox \cite{mitchell2005toolbox}, a Matlab \cite{MATLAB:2018} software package for solving partial differential equations using the level set method \cite{sethian1999level}. This set is visualized in Figure \ref{figure:frs} for $t=5.0$, and $x_{ext0} = \mathbf{0}$, where $\mathbf{0}$ is the zero vector. To compute the set $\tilde{\mathbf{D}}$, we first observe that the forward reachable set of the external system is invariant with respect to the initial condition $x_{ext0}$, up to a rigid transformation (rotation and translation). Therefore, we can solve for $\mathbf{R}_{ext}(t,\mathbf{0};0)$, and then compute the sets $\tilde{\mathbf{D}}(t,x_{ext0};0)$ and $\mathbf{A}_{int}(t, \tilde{\mathbf{D}},x_{ext0};0,5.0)$ in relative coordinates $x_1-x_4(0)$, $x_2-x_5(0)$, and $x_3-x_6(0)$. Here, we have avoided having to compute the sets for all possible values of $x_{ext0}$, as was discussed in Section \ref{sec:formulation}. 

To see how this is done, first consider  $\tilde{\mathbf{D}}(t,\mathbf{0};0)$.  We project the set $\mathbf{R}_{ext}(t,\mathbf{0};0)$ on to the  ($x_4$, $x_5$) plane, relabel the coordinates from $x_4 \to (x_1-x_4(0))$, $x_5 \to (x_2-x_5(0))$, and $x_6 \to (x_3-x_6(0))$, take the Minkowski sum of the projection with a circle of radius $2r$ (to account for the geometry of both systems), and then extrude the set in the $(x_3-x_6(0))$ dimension. The result of these operations are visualized in Figure \ref{figure:gtilde} for $t=5.0$. Finally, we compute the time-varying backwards avoid-tube of $\tilde{\mathbf{D}}(t,\mathbf{0};0)$, denoted $\mathbf{A}_{int}(t, \tilde{\mathbf{D}},\mathbf{0};0,5.0)$. This set is computed again using the Level Set toolbox, and is displayed in Figure \ref{figure:brt} at $t=0$.

Interpreting the set $\mathbf{A}_{int}(0, \tilde{\mathbf{D}},\mathbf{0};0,5.0)$, we see that it makes intuitive sense. The super-zero-level set (free space outside the green volume in Figure \ref{figure:brt}) is the set of all initial configurations of the internal and external systems, such that the internal system can escape collision despite adversarial control by the external system and without observations of the external system. The regions that are unable to avoid collision are those in which the initial configuration of the two systems are already in collision, or when the internal system is initially headed towards the external system, and cannot turn away in time due to the limited control authority on $u_2$ (note that the set wraps at $(x_3 - x_6(0)) = \pm \pi $ due to the periodicity of the relative heading). Because the bounds on $u_1$ are greater than those on $u_3$, the internal system can escape collision from all other initial configurations. 

While the result of the computations thus far in this example are sufficient to synthesize a controller which is able to avoid obstacles after observations are lost, it has been assumed that the obstacle to be avoided was observed before observations were lost. However, in dynamic environments, there are often situations in which no obstacles might be detected, such as when an obstacle is beyond the sensing range of the ego vehicle. Such a situation can also be accounted for using the method presented here. Specifically for this example, we first compute the set $\mathbf{R}_{ext}(t,\mathbf{S}_{sense};0)$, which is defined to be the forward reachable set of the external system when starting from any state beyond the ego-vehicle's sensing radius (relative to the origin): \begin{equation}
    \mathbf{S}_{sense} := \{\begin{bmatrix}x_4 & x_5 & x_6 \end{bmatrix}^\intercal : x_4^2 + x_5^2 \geq r_{sense}^2\}.
\end{equation} Here $r_{sense}$ is the sensing radius.

Code for generating the sets presented in this example is provided at \url{https://www.redacted_for_double_blind.review}.

\subsection{Notes on Computation} \label{sec:computation}

At the root of this method is the offline computation of the numerical solution of two partial differential equations. For PDEs without special structure, computing the numerical solution requires defining a grid over the state space and computing spatial derivatives at every grid point; the computation time thus scales exponentially with the dimension of the state space \cite{bansal2017hamilton}. In practice, we are limited in the dimension of systems for which we can compute these safety guarantees by the computation required to do so. Nevertheless, with clever representation of the systems, many interesting problems more complicated than the one presented here can be solved using this framework. 

Another noteworthy fact is that space discretization can introduce errors in the computation of the value functions of interest, although the discretization error tends to zero as bin size tends to zero \cite{evans1959pde}. Some artifacts from the discretization error can be seen in Figure \ref{figure:brt}, particularly in places where the set \textit{creases}. The possibility that these artifacts may appear should be taken into account when constructing safety arguments for systems using this method.

The computational burden at run-time is minimal, and requires simply looking up the pre-computed value and corresponding maximizing control from the grid, which is indexed by the configuration of the internal and external systems. The time taken by this operation is on the order of miliseconds.  

\section{Related Work}
The framework and method presented in this article stems most directly from the large body of work which has leveraged the Hamilton-Jacobi (HJ) equations to transform optimal control problems into numerically solvable partial differential equations. An overview of HJ optimal control formulations is given in \cite{bansal2017hamilton}. The Hamilton-Jacobi-Isaacs (HJI) equations similarly allow for transforming fully observable differential games into partial differential equations, as demonstrated in \cite{evans1984differential}, \cite{evans1959pde} \cite{fisac2015reach}. The work presented in this article is a differential game of sorts, although due to the impersistent observability of the external system, the standard HJI pursuit-evasion formulation does not apply, necessitating the method we have developed. The set of all safe initial configurations in the persistently observed case is a strict super-set of the eyes-closed safety kernel defined here. 

There are recent works that have leveraged HJ and HJI equations to capture other interesting information patterns arising in problems common to autonomous systems. For example, in \cite{herbert2017fastrack},  \cite{kousik2017safetrajectorysynthesis} and \cite{kousik2019safeaggressive}, a dynamic game is formulated to capture the mismatch between a motion planning model and a high-fidelity model of the system attempting to track the plans. This allows for planning trajectories that are guaranteed to be trackable by the high-fidelity model. These works are extended in \cite{fridovich2018planning} and \cite{kousik2018bridging}, allowing for limited sensing radii. Meanwhile, \cite{vaskov2019towardsprovablynotatfault} leverages a similar idea to generate trajectories that incorporate system tracking error and will not have any at-fault collisions in unforeseen environments with limited sensing radius. All of these methods assume that the internal system has persistent observations of all obstacles, at least within some sensing radius. The methods developed in these works do not apply when observation is lost. 

Other works involve different considerations of the perception system in autonomous vehicles. In \cite{falanga2019howfastistoofast}, for example, the authors consider the role of perception latency in safe avoidance of obstacles. Perhaps most similar to our work, from a motivation standpoint, is  \cite{Saxena2017LearningRF}, in which the authors attempt to learn a system that simultaneously detects perception failures and generates corrective action. However, that method does not include any formal safety analysis of the control actions that are generated. Furthermore, because such a system is learned from data, it  may unfortunately encounter the same failure modes that such a method attempts to be robust to in the first place. 

Recently, \cite{dean2019robust} formulated a robust control problem accounting for a learned perception module, although there the analysis requires that the dynamics of the autonomous system is a linear, time invariant ordinary differential equation. Furthermore, the high-dimensional inputs to the perception system, at test-time, are assumed to come from the same distribution of inputs at train-time. Therefore, the method presented there addresses a much more structured situation than the one considered here. 

It is also important to note that although we leverage exact reachable set computations to formulate our problem, techniques other than solving HJ PDEs can be used to find them. For example, \cite{frazzoli2002real} use RRT-type sampling to approximate such sets. Again, in that work, the obstacle avoidance and reachable sets considered do not represent the loss of observability situation we consider here. 

There are of course many works which have considered obstacle avoidance in the presence of noisy or uncertain measurements, such as \cite{du2011probabilistic}, \cite{althoff2012safety}, and \cite{Schmerling2017}, to cite just a few examples. Furthermore, there have been works which have tackled this problem in the context of formal methods, providing formal guarantees of safety in such observation models, as in \cite{mitsch2018verified}. However, again, consideration of noisy or purely partial observations is a different problem than the problem we address in this paper, and those works should be taken as complementary to the work presented here. 

\nocite{Mitchell2005}

\section{Conclusion}
We have presented a framework for optimally accounting for detected loss of observations of autonomous systems. We showed that in order to be robust to the situation defined, the control of the internal system must assume it will lose and never regain observations of the external system. Therefore, the internal system must always stay within a set of relative positions with respect to the external system, to ensure that even if the observations are in fact lost, there still exists a safe maneuver that avoids all possible trajectories of the external system. 

We first formulated this problem as an open-loop differential Stackelberg game, and then showed how it can be recast in a manner that allows its solution to be found by solving two related partial differential equations. 

This method results in a safe operating region for the autonomous system; here, we call it the Eyes-Closed Safety Kernel. If the autonomous system is able to detect instantaneously that it has lost observations, then by staying within the Eyes-Closed Safety Kernel for all time, and executing the control defined in (\ref{eq:opt_ctrl}) whenever observations are lost, the system is guaranteed to maintain safety as long as the system obeys the modeling assumptions made. 

Although the detection of perception failure is a very challenging problem in practice, the framework presented here demonstrates that it is possible to safely account for these failures if they can be detected. This method therefore motivates further research on developing better methods for predicting when perception modules are incorrect, and can already be used with existing methods for predicting perception failures. Even if not every perception failure is detected, the framework we present provides a means for guaranteeing robustness to those that are predicted. 










\bibliographystyle{./IEEEtran} 
\bibliography{root.bib}

\end{document}

%% file: notation.tex
\newcommand{\ra}{\rightarrow}
\newcommand{\R}{\mathbb{R}}

\newtheorem{theorem}{Theorem}

\newcommand{\runningexample}[1]%
{
\textbf{Running example:}
\textit{#1}
}